\documentclass[letterpaper, 10 pt, conference]{ieeeconf}

\IEEEoverridecommandlockouts
\usepackage{amsmath}
\usepackage{bbm}
\usepackage{amssymb}
\usepackage{graphicx}
\usepackage{dsfont}
\usepackage{xcolor}
\usepackage[noadjust]{cite}
\usepackage{hyperref}

\newtheorem{theorem}{Theorem}
\newtheorem{proposition}{Proposition}
\newtheorem{assumption}{Assumption}

\newcommand{\prt}[1]{\left(#1\right)}
\newcommand{\brk}[1]{\left[#1\right]}
\newcommand{\brc}[1]{\left\{#1\right\}}
\newcommand{\abs}[1]{\left|#1\right|}
\newcommand{\norm}[1]{\lVert#1\rVert}
\providecommand{\quotes[1]}{``#1"}
\newcommand{\E}{\mathbb{E}}
\newcommand{\R}{\mathbb R}
\newcommand{\Fab}{\mathcal F_{\alpha,\beta}}
\newcommand{\optx}[1]{x^{*,#1}}
\newcommand{\inProd}[2]{\langle#1,#2\rangle}
\DeclareMathOperator*{\argmin}{arg\,min}

\title{\LARGE \bf
Random coordinate descent algorithm for open multi-agent systems with complete topology and homogeneous agents
}

\author{Charles Monnoyer de Galland, Renato Vizuete, Julien M. Hendrickx, Paolo Frasca, and Elena Panteley
\thanks{Research supported by the “RevealFlight” ARC at UCLouvain, by the \textit{Incentive Grant for Scientific Research (MIS)} \quotes{Learning from Pairwise Data} of the F.R.S.-FNRS and in part by the Agence Nationale de la Recherche (ANR) via grant “Hybrid And Networked Dynamical sYstems” (HANDY), number ANR-18-CE40-0010.}
\thanks{C.~Monnoyer de Galland and R.~Vizuete equally contributed to this work. C.~Monnoyer de Galland and J. M. Hendrickx are with the ICTEAM institute, UCLouvain, Louvain-la-Neuve, Belgium. C.~Monnoyer de Galland is a FRIA fellow (F.R.S.-FNRS). R.~Vizuete and E.~Panteley are with Universit\'{e} Paris-Saclay, CNRS, CentraleSup\'{e}lec, Laboratoire des signaux et syst\`{e}mes, 91190, Gif-sur-Yvette, France. R.~Vizuete and P.~Frasca are with Univ.\ Grenoble Alpes, CNRS, Inria, Grenoble INP, GIPSA-lab, F-38000 Grenoble, France. (E-mail adresses: ~charles.monnoyer@uclouvain.be; ~renato.vizuete@l2s.centralesupelec.fr; ~julien.hendrickx@uclouvain.be; ~paolo.frasca@gipsa-lab.fr; ~elena.panteley@l2s.centralesupelec.fr).}}

\begin{document}

\maketitle
\thispagestyle{empty}
\pagestyle{empty}

\begin{abstract}
We study the convergence in expectation of the Random Coordinate Descent algorithm (RCD) for solving optimal resource allocations problems in open multi-agent systems, \emph{i.e.}, multi-agent systems that are subject to arrivals and departures of agents. 
Assuming all local functions are strongly-convex and smooth, and their minimizers lie in a given ball, we analyse the evolution of the distance to the minimizer in expectation when the system is
occasionally subject to replacements in addition to the usual iterations of the RCD algorithm.
We focus on complete graphs where all agents interact with each other with the same probability, and provide conditions to guarantee convergence in open system.
Finally, a discussion around the tightness of our results is provided. 
\end{abstract}

\section{Introduction}

We consider the optimal resource allocation problem stated as follows, where a budget $b\in\R^d$ must be distributed among $n$ agents according to some weight distribution $a\in\R^n$ while minimizing the total cost $f$ built upon local costs $f_i:\R^d\to\R$ (the weights $a_i$ are thus scalar):
\begin{align}
    \label{eq:Statement:ResourceAllocProblem}
    \min_{x\in\R^n}\ \ &f(x) = \sum_{i=1}^n f_i(x_i)
    &\hbox{subject to}& &\sum_{i=1}^na_ix_i=b.
\end{align}

Such problems arise in different fields of research, including power systems \cite{yi2016initialization}, actuator networks \cite{teixeira2013distributed}, and games \cite{liang2017distributed}. 
Some of the first approaches introduced to solve \eqref{eq:Statement:ResourceAllocProblem} rely on distributed algorithms based on the well known \textit{gradient descent} \cite{xiao2006optimal}. 
Algorithms of this type however require computing the full gradient of the network, such that the computational complexity can be too high for large systems.

To reduce the computational complexity of gradient-based algorithms, Nesterov introduced in \cite{nesterov2012efficiency} the \emph{coordinate descent} algorithm where optimization steps are performed along only one direction at each iteration.
Several extensions of this algorithm have been developed, including a block coordinate update \cite{richtarik2014iteration}, where more than one direction is optimized at each iteration. 
In such algorithms, the sequence of coordinates in which updates are performed plays an important role, and it is well-known that randomized choices can guarantee convergence.
Hence, \cite{necoara2013random} proposed a \emph{random coordinate descent} (RCD) algorithm, where at each iteration only a pair of local gradients must be evaluated, and where that pair is randomly selected, guaranteeing convergence at the same time as reducing computational complexity.

In some applications of \eqref{eq:Statement:ResourceAllocProblem}, agents are able to join and leave the system at a time-scale similar to that of the process.
Consider for instance the integration of distributed energy resources \cite{dominguez2012decentralized}, where some devices (agents) supplying a total amount of resource (budget) can sometimes be unavailable because of a fault or where local objective-functions might be time-varying (\textit{e.g.}, due to environmental conditions for photovoltaic systems).
When the size of the system increases, the probability for such perturbations to happen at the scale of the whole system increases as well, giving rise to optimization problems in \emph{open multi-agent systems}.
In that case, arrivals and departures have a significant effect on the course of algorithms
and even the most basic algorithms fail to guarantee convergence due to the continuous change of the set of agents. 
In particular, arrivals and departures of agents result in variations of the cost functions during the process, and hence of the location of the minimizer as well, which prevent convergence.

Motivated by the possible changes of the functions $f_i$ in \eqref{eq:Statement:ResourceAllocProblem}, we analyze the performance of the RCD algorithm introduced in \cite{necoara2013random} in a system subject to possible replacements of cost functions at each iteration.
We extend the results of \cite{necoara2013random} by analysing the convergence rate in expectation of the distance to the minimizer in open systems, under the assumption that each iteration is either an RCD update or a replacement.
In this work, we focus on complete graphs such that each pair of agents updates its state at some iteration with the same probability, and we assume that the local objective functions are smooth and strongly convex. 
We then analyze the tightness of our results by considering the particular case of quadratic cost functions, and relying on the PESTO toolbox \cite{PESTO}, which allows deriving exact empirical bounds for convex problems.

\subsection{State of the art}
In the last years, traditional algorithms have been applied and analysed in open multi-agent systems, such as gossiping in \cite{de2020open,hendrickx2017open,OMAS:ARXIV_ITAC:FPL:2020}, dynamic consensus in \cite{franceschelli2020stability,dashti2019dynamic}, and stochastic interactions in \cite{vizuete2020influence,varma2018open}.
Optimization in open system is also getting attention, such as in \cite{hsieh2021optimization}, or in \cite{OpenDo:OpenDGDStability} where the authors studied the stability of the decentralized gradient descent algorithm where the agents try to reach agreement and can be replaced at each iteration.

An alternative line of work on time-varying objective functions, called \emph{online optimization} \cite{DO:online-varyingFunctions,shahrampour2017distributed}, aims at building at each time $t$ an estimate $x^t$ in a way that keeps the \emph{regret function}, commonly defined as
\begin{equation}
    \label{Intro:OCO:RegretAVG}
    Reg_T := \sum\nolimits_{t = 1}^T \prt{f^t(x^t) - \min_x f^t(x)},
\end{equation}
as small as possible. 
Nevertheless, our work is essentially different because the objective of our algorithms is to be at all times as close as possible to the instantaneous minimizer of \eqref{eq:Statement:ResourceAllocProblem}.

\section{Problem statement}
\label{Sec:Statement}

For two vectors $x,y\in\R^n$, we denote by $\inProd{x}{y} = x^\top y = \sum_{i=1}^n x_iy_i$ the standard Euclidean inner product, and the Euclidean norm by $\norm{x} = \prt{x^\top x}^{1/2}$.
We also denote the vector of size $n$ constituted of only ones by $\mathds{1}_n$ and the identity matrix of dimension $n$ by $I_n$.
Let $B(x,r)=\{y:\norm{x-y}\le r \}$ denote the ball of radius $r\ge 0$ centered at $x$.

\subsection{Resource allocation problem}
\label{Sec:Statement:ResourceAlloc}

We consider the resource allocation problem \eqref{eq:Statement:ResourceAllocProblem} where we restrict our attention to 1-dimensional local cost functions $f_i:\R\to\R$ for all $i=1,\ldots,n$, and make the following classical assumption. 
\begin{assumption}
\label{Ass:Statement:Fab}
    Each function $f_i$ is continuously differentiable, $\alpha$-strongly convex (\emph{i.e.}, $f_i(x)-\frac{\alpha}{2}\norm{x}^2$ is convex) and $\beta$-smooth (\emph{i.e.}, $\norm{ f_i'(x)- f_i'(y)}\leq \beta\norm{x-y}$, $\forall x,y$).
\end{assumption}

We let $\Fab$ denote the set containing the functions satisfying Assumption~\ref{Ass:Statement:Fab} let $\kappa=\beta/\alpha$ denote the \emph{condition number} of those functions.
Notice that $f(x) := \sum_{i=1}^n f_i(x_i)$ also satisfies Assumption~\ref{Ass:Statement:Fab}, so that $f\in\Fab$.
This implies that the solution to \eqref{eq:Statement:ResourceAllocProblem}, denoted $x^*$, is unique. 
Moreover, 
$\inProd{a}{x^*}=b$ and $\nabla f(x^*) = \lambda^*a$ for some scalar $\lambda^*\in\R$.

In open systems, the functions $f_i$ can be replaced in the process so that the global minimizer $x^*$ changes along. To ensure that the local cost functions are consistent with each other, and prevent arbitrary changes of functions, and thus of $x^*$, we follow the approach in \cite{OpenDo:OpenDGDStability} and restrict the location of the local minimizers without loss of generality.
\begin{assumption}
\label{Ass:Statement:B(0,1)}
    The minimizer of each function $f_i$ denoted $x_i^*:=\argmin_{x}f_i(x)$ satisfies $x_i^*\in [-1,1]$ and $f_i(x_i^*)=0$.
\end{assumption}

The following assumption restricts our attention to the particular case where a given budget must be allocated among agents with the same priority, or where the budget is provided by a group of homogeneous agents.

\begin{assumption}
\label{Ass:Statement:a=1}
There holds $a = \mathds{1}_n$, and we denote the feasible set of \eqref{eq:Statement:ResourceAllocProblem} in that case by
\begin{equation}
    \label{e:Statement:Sb}
    S_b := \brc{x\in\R^n|\inProd{\mathds{1}}{x}=b}.
\end{equation}
\end{assumption}

\subsection{Random Coordinate Descent algorithm}
\label{Sec:Statement:RCD}

To problem \eqref{eq:Statement:ResourceAllocProblem}, we associate a network constituted of $n$ agents such that each agent $i\in V = \brc{1,\ldots,n}$ has access to a local function $f_i$ and a local variable $x_i\in\R$.
The agents can exchange information according to an undirected and connected graph $G=(V,E)$ where $E\subseteq V\times V$.

The Random Coordinate Descent (RCD) algorithm introduced in \cite{necoara2013random} involves the update of the states of only a pair of neighbouring agents at each iteration, so that the numerical complexity is cheap.
At a given iteration and for some feasible estimate $x$, a pair of agents $(i,j)\in E$ is randomly selected with probability $p_{ij}>0$ to update as
\begin{align*}
    &x_i^+ = x_i + d_i&
    &x_j^+ = x_j + d_j,
\end{align*}
where $d_i$ and $d_j$ are determined by solving
\begin{small}
\begin{equation}
\label{eq:Statement:RCD:minProb}
    \begin{bmatrix}
        d_i\\d_j
    \end{bmatrix}
    = \arg\min_{s:a_is_i+a_js_j=0}
    \left\langle
    \begin{bmatrix}
        f'_i(x_i)\\ f'_j(x_j)
    \end{bmatrix},
    \begin{bmatrix}
        s_i\\s_j
    \end{bmatrix}\right\rangle
    + \frac\beta2\left\lVert
    \begin{bmatrix}
        s_i\\s_j
    \end{bmatrix}
    \right\lVert^2.
\end{equation}
\end{small}

This choice follows the observation that for any $z\in\R^2$, the function $g(z) = f_i(z_1) + f_j(z_2)$ is $\beta$-smooth, and thus satisfies by definition  $\forall z,w\in\R^2$
\begin{equation}
\label{eq:Statement:betaProperty}
    g(z) \leq g(w) + \inProd{\nabla g(w)}{z-w} + \frac\beta2\norm{z-w}^2.
\end{equation}
Solving \eqref{eq:Statement:RCD:minProb} thus amounts to minimizing the right hand side of \eqref{eq:Statement:betaProperty} while ensuring that the next estimate $x^+$ is still feasible.
Following the approach in \cite{necoara2013random}, the problem is solved by
\begin{equation}
\label{eq:Statement:RCD:dClosedForm}
    \begin{bmatrix}
        d_i\\d_j
    \end{bmatrix}
    = -\frac{1}{\beta}
    \begin{bmatrix}
        1-\frac{a_i^2}{a_i+a_j}     &-\frac{a_ia_j}{a_i^2+a_i^2}\\
        -\frac{a_ia_j}{a_i^2+a_i^2}
        &1-\frac{a_j^2}{a_i+a_j} 
    \end{bmatrix}
    \begin{bmatrix}
        f'_i(x_i)\\f'_j(x_j)
    \end{bmatrix}.
\end{equation} 
Under Assumption~\ref{Ass:Statement:a=1}, one gets the following iteration rule
\begin{equation}
\label{eq:Statement:RCD:UpdateRule}
    x^+ = x - \frac{1}{\beta}Q^{ij}\nabla f(x),
\end{equation}
where $Q^{ij}$ is a $n \times n$ matrix filled with zeroes except for the four following entries
\begin{align*}
    &[Q^{ij}]_{i,i} = [Q^{ij}]_{j,j} = \frac12;&
    &[Q^{ij}]_{i,j} = [Q^{ij}]_{j,i} = -\frac12.
\end{align*}

In this preliminary work, we restrain to fully connected networks as in \cite{de2020open,hendrickx2017open}. We have thus all-to-all (possible) communications and each edge has the same probability to be selected at an iteration of the RCD algorithm.
\begin{assumption}
\label{Ass:Statement:CompleteGraph}
    The graph $G=(V,E)$ is fully connected, and for all $(i,j)\in E$ there holds $p_{ij} = p = \frac{2}{n(n-1)}$.
\end{assumption}
Hence, under Assumption~\ref{Ass:Statement:CompleteGraph}, there holds
\begin{equation}
    \label{eq:Statement:LwrtQij}
    \sum_{(i,j)\in E} p_{ij}Q^{ij} = \frac{p}{2}L,
\end{equation}
where $L$ is the Laplacian matrix of $G$, given by \begin{equation}
    \label{eq:Statement:L}
    L = nI_n-\mathds 1_n\mathds 1_n^\top .
\end{equation}

\subsection{Function replacement}
\label{Sec:Statement:Objective}
In this analysis, we consider that the system is open. 
In particular, any agent $i$ can be replaced during the process, in which case it receives a new local objective function satisfying Assumptions~\ref{Ass:Statement:Fab} and \ref{Ass:Statement:B(0,1)} and maintains its label and estimate so that $\inProd{a}{x}=b$ is preserved. 
Let $f_i^k$ denote the local objective function held by the agent labelled $i$ at iteration $k$, then \eqref{eq:Statement:ResourceAllocProblem} can be reformulated in our setting as
\begin{equation}
    \label{eq:Statement:Objective:OpenResourceAllocProblem}
    \min_{x\in S_b} f^k(x) := \sum_{i=1}^n f_i^k(x_i).
\end{equation}
The solution of \eqref{eq:Statement:Objective:OpenResourceAllocProblem} thus changes with replacements, and we denote $\optx{k} := \argmin_{x\in S_b} f^k(x)$.
Let $x^k$ be the estimate of $\optx{k}$ at iteration $k$, we define the following error metric:
\begin{equation}
    \label{eq:Statement:Objective:Ck}
    C^k := \norm{x^k-\optx{k}}^2.
\end{equation}
Our goal is to derive a convergence rate for criterion \eqref{eq:Statement:Objective:Ck} in expectation given by $\E\brk{C^k}$, where $x^k$ is a sequence generated by the Random Coordinate Descent algorithm \eqref{eq:Statement:RCD:UpdateRule} applied in a system subject to possible replacements of agents.

\section{Convergence of RCD in closed system}

In this section we analyze the convergence rate in expectation of the RCD algorithm for criterion \eqref{eq:Statement:Objective:Ck} in closed system for our setting.
In that case, the minimizer $\optx{k}$ does not depend on $k$, since the local objective functions $f_i$ remain the same during the process. Therefore, we refer to that minimizer as $x^*$ in this section.

A related result was presented in \cite{necoara2013random}, where such convergence rate in expectation was derived for the objective value $f(x)-f(x^*)$.
Proposition~\ref{Prop:EConvRate:updates} is thus an extension of that result for our metric, and will serve as an intermediate result for working on open systems. 
Interestingly, one can show that while $f(x)-f(x^*)$ is always decreasing, the metric $\norm{x-x^*}$ can increase for certain choices of edges.
We consider the following iteration rule, which is a generalization of \eqref{eq:Statement:RCD:UpdateRule} for general positive step-sizes $h$:
\begin{equation}
    \label{eq:Prop:EConvRate:updates:rule}
    x^+ = x - hQ^{ij}\nabla f(x).
\end{equation}

\begin{proposition}
\label{Prop:EConvRate:updates}
    Let a function $f(x):= \sum_{i=1}^n f_i(x_i)$ and $x^*:= \argmin_{x\in S_b} f(x)$.
    Under Assumptions~\ref{Ass:Statement:Fab} to \ref{Ass:Statement:CompleteGraph}, for any positive scalar $h\leq 1/\beta$, and for any initial point $x\in S_b$, then the update rule \eqref{eq:Prop:EConvRate:updates:rule}
    applied on the randomly selected pair of agents $(i,j)\in E$ satisfies
    \vspace{-0.1cm}
    \begin{equation}
        \label{eq:Prop:EConvRate:updates:h}
        \E\brk{\norm{x^+-x^*}^2} 
        \leq \prt{1 - h\frac{\alpha}{n-1}}\norm{x-x^*}^2.
    \end{equation}
\end{proposition}

\begin{proof}
    Starting from the update rule \eqref{eq:Prop:EConvRate:updates:rule}, there holds
    \begin{align*}
        \E\brk{\norm{x^+-x^*}^2} 
        &= \sum_{(i,j)\in E} p \E\brk{\norm{x - hQ^{ij}\nabla f(x) - x^*}^2}\\
        &= \sum_{(i,j)\in E} p \norm{x - hQ^{ij}\nabla f(x) - x^*}^2.
    \end{align*}
    Let {\small $H(x) = \E\brk{\norm{x^+-x^*}^2}$}. It follows that
    \begin{align*}
        H(x)
        = &\norm{x-x^*}^2
        + h^2\sum\nolimits_{(i,j)\in E} p\norm{Q^{ij}\nabla f(x)}^2\\
        &- 2h\sum\nolimits_{(i,j)\in E} p\inProd{Q^{ij}\nabla f(x)}{x-x^*}.
    \end{align*}
    From \eqref{eq:Statement:LwrtQij}, there holds $\sum_{(i,j)\in E} Q^{ij} = \tfrac12L$.
    Moreover, one has $(Q^{ij})^\top=Q^{ij}$ and $\prt{Q^{ij}}^2=Q^{ij}$, and it follows
    \begin{small}
    \begin{align*}
        \sum_{i=1}^n p\norm{Q^{ij}\nabla f(x)}^2
        &= p\nabla f(x)^\top \prt{\sum_{(i,j)\in E}(Q^{ij})^2} \nabla f(x)\\
        &= \frac{p}{2}\nabla f(x)^\top L\nabla f(x);\\
        \sum_{i=1}^n p\inProd{Q^{ij}\nabla f(x)}{x-x^*}
        &= p\nabla f(x)^\top \prt{\sum_{(i,j)\in E} Q^{ij}}(x-x^*)\\
        &= \frac{p}{2}\nabla f(x)^\top L(x-x^*).
    \end{align*}
    \end{small}
    Hence, using $p = \frac{2}{n(n-1)}$ from Assumption~\ref{Ass:Statement:CompleteGraph}, there holds
    \begin{align*}
        H&(x)
        = \norm{x-x^*}^2 \\
        &+ \frac{1}{n-1}\prt{\tfrac{h^2}{n}\nabla f(x)^\top L\nabla f(x) - 2\tfrac{h}{n}\nabla f(x)^\top L(x-x^*)}.
    \end{align*}
    The optimality conditions of our problem imply $\nabla f(x^*)=\lambda^*\mathds 1_n$ for some $\lambda^*\in\R$, and from \eqref{eq:Statement:L} we have
    $$L\nabla f(x^*) = \mathds 1_nn\lambda^* - \mathds 1_n \mathds 1_n^\top  \mathds 1_n \lambda^* = 0.$$
    Hence, since $L\nabla f(x^*)=0$, since $L=L^\top $, and since the largest eigenvalue of $L$ is $n$, there holds
    \begin{small}
    \begin{align*}
        \tfrac{h^2}{n}\nabla f(x)^\top \!L\nabla f(x)\!
        &=\! \tfrac{h^2}{n}(\nabla f(x)\!-\!\nabla f(x^*))^\top \!L(\nabla f(x)\!-\!\nabla f(x^*))\\
        &\leq h^2\norm{\nabla f(x)-\nabla f(x^*)}^2.
    \end{align*}
    \end{small}
    Moreover, one has $\inProd{\mathds 1_n}{x-x^*}=0$, and using \eqref{eq:Statement:L} yields
    \begin{align*}
        \tfrac{h}{n}\nabla f(x)^\top L(x-x^*)
        &= h\inProd{\nabla f(x) - \tfrac{1}{n}\mathds 1_n^\top \nabla f(x)\mathds 1_n}{x-x^*}\\
        &= h\inProd{\nabla f(x)}{x-x^*}.
    \end{align*}
    Furthermore, since $\inProd{\nabla f(x^*)}{x-x^*}=0$, and since $f$ is $\alpha$-strongly convex and $\beta$-smooth, it follows that
    \begin{align*}
        \inProd{\nabla f(x)}{x-x^*}
        &= \inProd{\nabla f(x)-\nabla f(x^*)}{x-x^*}\\
        &\geq \tfrac{\beta^{-1}\norm{\nabla f(x)-\nabla f(x^*)}^2}{1+\kappa^{-1}} + \tfrac{\alpha\norm{x-x^*}^2}{1+\kappa^{-1}},
    \end{align*}
    where we remind $\kappa=\beta/\alpha$ is the condition number of $f$.
    Re-injecting those expressions into that of $H$ yields
    \begin{align*}
        H(x)
        \leq &\ \norm{x-x^*}^2-\frac{1}{n-1}\prt{\tfrac{2h\alpha}{1+\kappa^{-1}}\norm{x-x^*}^2}\\
        & + \frac{1}{n-1}\prt{\prt{h^2-\tfrac{2h\beta^{-1}}{1+\kappa^{-1}}}\norm{\nabla f(x)-\nabla fx^*)}^2}.
    \end{align*}
    Observe that for $h\leq 1/\beta$, we have
    \begin{align*}
        h^2-\tfrac{2h\beta^{-1}}{1+\kappa^{-1}}\leq 0,
    \end{align*}
    so that for any $h\leq 1/\beta$ there holds
    \begin{align*}
        H(x) \leq \prt{1-\frac{2}{1+\kappa^{-1}}\frac{\alpha h}{n-1}}\norm{x-x^*}^2.
    \end{align*}
    Finally, since $\frac{2}{1+\kappa^{-1}} \geq 1$, there holds
    \begin{align*}
        \E\brk{\norm{x^+-x^*}^2} 
        = H(x)
        \leq \prt{1-\frac{\alpha h}{n-1}}\norm{x-x^*}^2,
    \end{align*}
    which concludes the proof.
\end{proof}

Observe that for $h = 1/\beta$, the iteration rule \eqref{eq:Prop:EConvRate:updates:rule} corresponds to that of the RCD algorithm given in \eqref{eq:Statement:RCD:UpdateRule}, which yields the following convergence rate in closed system
\begin{equation}
    \label{eq:EConvRate:updates:h=1/beta}
    \E\brk{\norm{x^+-x^*}^2} \leq
    \prt{1-\tfrac{1}{(n-1)\kappa}}\norm{x-x^*}^2.
\end{equation}
This also corresponds to the contraction rate observed in an open system upon one iteration where no replacement takes place.
Observe moreover that this rate is linear, and similar to that of a gradient descent algorithm \cite{nesterov2018lectures,MAL-050}.

\section{Convergence of RCD in open system}
\label{Sec:OpenSystems}

We now consider that the system is open and suffers from occasional replacements of agents so that the local objective functions $f_i$ change. 
In particular, when a replacement occurs, then the replaced agent $i$ is uniformly randomly selected, and receives a new objective function $f_i$ satisfying Assumptions~\ref{Ass:Statement:Fab} and \ref{Ass:Statement:B(0,1)} while maintaining its estimate.

Let $U_{ij}$ denote the event that an RCD iteration as defined in \eqref{eq:Statement:RCD:UpdateRule} happens on the pair of agents $(i,j)$, and let $R_i$ denote the event of a replacement of agent $i$ as described above.
Then we define the set of all possible events as 
\begin{equation}
    \label{eq:Statement:EventSet}
    \Xi = \prt{\bigcup_{(i,j)\in E}U_{ij}} \cup \prt{\bigcup_{i\in V} R_i}.
\end{equation}

We consider that at each iteration one event $\xi\in\Xi$ takes place, so that we can define the history of the process up to iteration $k$ as follows:
\begin{equation}
    \label{eq:Statement:History}
    \omega^k = \brc{(1,\xi_1),\ldots,(k,\xi_k)},
\end{equation}
with $\xi_j\in\Xi$ for all $j=1,\ldots,k$.
We will work under the following assumption of statistical independence:
\begin{assumption}
\label{Ass:OpenSystems:Indep}
The events $\xi_i$ constituting any sequence of events $\omega^k$ are independent of each other and of the state of the system, so that at any iteration $i$, the event $\xi_i$ is a RCD update with probability $p_U$, and a replacement with probability $p_R = 1-p_U$.
\end{assumption}
The assumption above guarantees that the replacements and RCD updates happening in the system are independent processes, and allows analyzing the behavior of the RCD algorithm by decoupling the impact of these.
In the remainder of this section, we will analyze the convergence rate of the algorithm by analyzing separately the effect of updates of the algorithm and of replacements on the error metric \eqref{eq:Statement:Objective:Ck}.

Observe that the probabilities $p_U$ and $p_R$ act at the whole system level, and can equivalently be replaced by the corresponding probabilities acting on every single agent and edge on the system. 
In particular, it follows from Assumption~\ref{Ass:OpenSystems:Indep} 
\begin{align}
    \label{eq:OpenSystems:pepaWRTpUpR}
    &p_e = \tfrac{2}{n(n-1)}p_U&
    &p_a = \tfrac1np_R,
\end{align}
where $p_e$ and $p_a$ respectively stand for the probabilities that any given edge gets activated at a RCD update, and that any given agent (whichever it is) is replaced at some iteration.

\subsection{Impact of replacements on the error}
In this section, we analyze how much the minimizer $\optx{k}$ of Problem \eqref{eq:Statement:Objective:OpenResourceAllocProblem} is impacted by replacements, and to what extent the error $\E\brk{C^k}$ is affected by these.
Observe that the way we model replacements legitimates the analysis of the effect of a single change, as only one replacement at most can occur at a given iteration.

We first provide in the next proposition the region in which that minimizer can be located in our setting.

\begin{proposition}
\label{prop:OpenSystems:minLoc}
    Let $\kappa = \beta/\alpha$ denote the condition number of $f$, and let $R_{b,\kappa}:=\sqrt{n}+\prt{1+\frac{\abs{b}}{n}}\sqrt{\kappa n}$. If $f_i$ satisfies Assumptions~\ref{Ass:Statement:Fab} and \ref{Ass:Statement:B(0,1)} for all $i=1,\ldots,n$, then:
    \begin{equation}
        \label{eq:prop:minLoc}
        \argmin_{x\in S_b} f(x)\in B(0,R_{b,\kappa}).
    \end{equation}
\end{proposition}
\begin{proof}
Let $x\notin B(0,R_{b,\kappa})$ such that $x\in S_b$, and let $\bar x^* = \argmin_x f(x)$ denote the minimizer of $f$ without the constraint. 
Observe that from Assumption~\ref{Ass:Statement:B(0,1)}, there holds $f(\bar x^*)=0$ since it amounts to evaluating every local function $f_i$ at their minimal values.
Moreover, we have $\bar x^* \in B(0,1)^n$ so that $\norm{\bar x^*}\leq \sqrt{n}$, and it follows that $\norm{x-\bar x^*}>R_{b,\kappa}-\sqrt{n}$.
Hence, since $f$ is $\alpha$-strongly convex, there holds
\begin{equation*}
    f(x) 
    \geq \frac\alpha2\norm{x-\bar x^*}^2
    > \tfrac\alpha2\prt{1+\tfrac{\abs{b}}{n}}^2\kappa n
    = \tfrac{\beta n}{2}\prt{1+\tfrac{\abs{b}}{n}}^2.
\end{equation*}
Now let $x_b := \tfrac{b}{n}\mathds 1_n$.
Since $f$ is $\beta$-smooth, and since $f(\bar x^*)=0$ from Assumption~\ref{Ass:Statement:B(0,1)}, there holds
\begin{align*}
    f(x_b) \leq \tfrac{\beta}{2}\norm{x_b-\bar x^*}^2 
    \leq \tfrac{\beta n}{2} \prt{1+\tfrac{\abs{b}}{n}}^2.
\end{align*}

Hence, since $x_b\in S_b$, there holds
\begin{align*}
    f(x) > \tfrac{\beta n}{2}\prt{1+\tfrac{\abs{b}}{n}}^2 
    \geq f(x_b) 
    \geq f(x^*),
\end{align*}
and we conclude that $x$ cannot be the minimizer of \eqref{eq:Statement:Objective:OpenResourceAllocProblem}.
\end{proof}

We can now analyze the impact of a function change on the location of the minimizer.
Without loss of generality, we assume that the function that gets replaced is $f_n$, and for the $n+1$ functions $f_1,f_2,\ldots,f_{n-1},f_n^{(1)},f_n^{(2)}$ satisfying Assumptions~\ref{Ass:Statement:Fab} and \ref{Ass:Statement:B(0,1)}, we define
\begin{align}
    \label{eq:OpenSystems:x+&x-}
    x^{(1)} &:= \argmin_{x\in S_b} \prt{\sum\nolimits_{i=1}^{n-1}f_i(x_i) + f_n^{(1)}(x_n)};\nonumber\\
    x^{(2)} &:= \argmin_{x\in S_b} \prt{\sum\nolimits_{i=1}^{n-1}f_i(x_i) + f_n^{(2)}(x_n)}.
\end{align}
We provide in the next proposition an upper bound on $\norm{x^{(2)}~-~x^{(1)}}^2$, built upon Proposition~\ref{prop:OpenSystems:minLoc}.

\begin{proposition}
\label{Prop:OpenSystems:minChange}
    Consider $x^{(1)}$ and $x^{(2)}$ from \eqref{eq:OpenSystems:x+&x-}, then
    \begin{equation}
        \label{eq:Prop:OpenSystems:minChange}
        \norm{x^{(2)}-x^{(1)}}^2\leq        4n\kappa\prt{1+\tfrac{1}{\sqrt{\kappa}}+\tfrac{\abs{b}}{n}}^2.
    \end{equation}
\end{proposition}

\begin{proof}
    From Proposition~\ref{prop:OpenSystems:minLoc}, the minimizer $\optx{k}$ of Problem \eqref{eq:Statement:Objective:OpenResourceAllocProblem} satisfies
    \begin{align*}
        \norm{\optx{k}}^2 \leq n\prt{1+\prt{1+\tfrac{\abs{b}}{n}}\sqrt{\kappa}}^2.
    \end{align*}
    Hence, the conclusion follows from
    \begin{align*}
        \norm{x^{(2)}-x^{(1)}}^2 \leq 2\prt{\norm{x^{(2)}}^2+\norm{x^{(1)}}^2},
    \end{align*}
    as both $x^{(2)}$ and $x^{(1)}$ are such minimizers.
\end{proof}

The bound obtained in Proposition~\ref{Prop:OpenSystems:minChange} builds on the possibility for all agents to be replaced at once in a single iteration.
As a consequence, the results we derive using it are valid for that more general setting.
This also means that this result is a source of conservatism in the particular setting where only one agent can get replaced at a time, and it is expected that a tighter bound can be obtained in that case, especially regarding its dependence in $n$.
This possibility is discussed in detail in Section~\ref{Sec:OpenSystems:QuadFun}, through the study of a specific case, and based on the PESTO toolbox for performance estimation \cite{PESTO}.
However, the analysis in general remains open shall be the object of future work.

We can now evaluate the effect of replacements on the expected error $\E\brk{C^k}$.

\begin{proposition}
\label{Prop:OpenSystems:EConvRate:repl}
    Let $R$ denote the event of a replacement happening in the system. Then there holds
    \begin{equation}
        \label{eq:Prop:OpenSystems:EConvRate:repl}
        \E\brk{C^{k+1}|R}
        \leq 2\E\brk{C^k}+8n\kappa\prt{1+\tfrac{1}{\sqrt{\kappa}}+\tfrac{\abs{b}}{n}}^2.
    \end{equation}
\end{proposition}

\begin{proof}
    Let us fix some event sequence $\omega^{k-1}$. Using Assumption~\ref{Ass:OpenSystems:Indep}, there holds
    \begin{align*}
        \E\brk{C^{k+1}|R,\omega^{k-1}} = \sum_{i=1}^n p_i\E\brk{C^{k+1}|R_i,\omega^{k-1}},
    \end{align*}
    where $p_i$ is the probability that agent $i$ is the replaced agent at the occurrence of a replacement.
    
    Let $\optx{k}$ denote the minimizer of \eqref{eq:Statement:Objective:OpenResourceAllocProblem} before the replacement, so that $C^k = \norm{x^k-\optx{k}}^2$.
    In the event $R_i$, the estimates satisfy $x^{k+1}=x^k$, and there holds
    \begin{align*}
        C^{k+1}
        &= \norm{x^{k+1}-\optx{k+1}}^2\\
        &\leq \prt{\norm{x^k-\optx{k}} + \norm{\optx{k}-\optx{k+1}}}^2\\
        &\leq 2\prt{C^k + \norm{\optx{k}-\optx{k+1}}^2},
    \end{align*}
    where we have used the fact that $(a+b)^2\leq 2a^2+2b^2$ for $a,b\in\R$ to obtain the last inequality.
    It then follows from Proposition~\ref{Prop:OpenSystems:minChange} that
    \begin{align*}
        \norm{\optx{k}-\optx{k+1}}^2
        \leq 4n\kappa\prt{1+\tfrac{1}{\sqrt{\kappa}}+\tfrac{\abs{b}}{n}}^2,
    \end{align*}
    so that
    \begin{align*}
        \E\brk{C^{k+1}|R_i,\omega^{k-1}}
        \leq 2 C^k+8n\kappa\prt{1+\tfrac{1}{\sqrt{\kappa}}+\tfrac{\abs{b}}{n}}^2.
    \end{align*}
    The conclusion then follows from $p_i=1/n$ for all $i$ by definition and from taking the expectation over $\omega^{k-1}$.
\end{proof}

\subsection{Convergence rate}
\label{Sec:OpenSyetem:ConvRates}

We now analyze the convergence in expectation of the RCD algorithm when the system is subject to replacements.
Relying on the definition of the replacement process, our approach allows decoupling the effects of the algorithm and of replacements by considering that either a replacement or an update of the algorithm happens at each iteration.
Therefore, our results strongly depend on the analysis of the effect of replacement events obtained in the previous section.
Moreover, our methodology can be extended to different algorithms than the Random Coordinate Descent, as the impact of function changes is independent of the algorithm.

We provide in the following theorem a convergence rate in expectation for our error metric \eqref{eq:Statement:Objective:Ck} in a  system subject to replacements.

\begin{theorem}
\label{Thm:OpenSystems:ConvRate:generalFi}
    Under Assumptions~\ref{Ass:Statement:Fab} to \ref{Ass:OpenSystems:Indep}, the iteration rule \eqref{eq:Statement:RCD:UpdateRule} applied on a system subject to replacements generates a sequence of estimates $x^k$ satisfying for all $k\geq0$
    \begin{align}
        \label{eq:Thm:OpenSystems:ConvRate:generalFi}
        \E\brk{C^{k+1}} \leq \prt{2-p_U\prt{1+\tfrac{1}{(n-1)\kappa}}}\E\brk{C^k} + \Gamma,
    \end{align}
    with
    \begin{equation}
        \label{eqThm:OpenSystems:ConvRate:generalFi:Gamma}
        \Gamma =         8(1-p_U)\prt{1+\tfrac{1}{\sqrt{\kappa}}+\tfrac{\abs{b}}{n}}^2n\kappa.
    \end{equation}
\end{theorem}

\begin{proof}
    Let $U$ and $R$ respectively denote the occurence of a RCD update and of a replacement.
    There holds
    \begin{align*}
        \E\brk{C^{k+1}}
        &= p_U\E\brk{C^{k+1}|U} + p_R\E\brk{C^{k+1}|R},
    \end{align*}
    where we remind $p_U$ stands for the probability that an event is a RCD iteration, and $p_R$ the complementary probability that an event is a replacement, so that $p_U+p_R=1$.
    
    The first term corresponds to the convergence rate in expectation of a RCD iteration in closed system with a step-size of $1/\beta$. Hence, from Proposition~\ref{Prop:EConvRate:updates}, there holds
    \begin{align*}
        \E\brk{C^{k+1}|U}
        &\leq \prt{1-\tfrac{1}{(n-1)\kappa}}\E\brk{C^k}.
    \end{align*}
    
    Similarly, the second term is obtained from Proposition~\ref{Prop:OpenSystems:EConvRate:repl} and there holds
    \begin{align*}
        \E\brk{C^{k+1}|R} 
        \leq 2\E\brk{C^k}+8n\kappa\prt{1+\tfrac{1}{\sqrt{\kappa}}+\tfrac{\abs{b}}{n}}^2.
    \end{align*}
    
    Combining those expressions, and using the fact that $p_U+p_R=1$ concludes the proof.
\end{proof}

The convergence rate obtained in Theorem~\ref{Thm:OpenSystems:ConvRate:generalFi} allows upper bounding the performance of the RCD algorithm under replacements events.

First observe that convergence is guaranteed as long as the probability for an event to be an RCD update $p_U$ satisfies
\begin{align}
    \label{eq:OpenSystems:ConvGuarantee:pU}
    p_U > \frac{\kappa(n-1)}{\kappa( n-1)+1},
\end{align}
which corresponds to the worst-case contraction rate guaranteeing contraction in expectation at each iteration.

Let us denote $\rho_R := p_R/p_U$ the expected number of replacements happening between two RCD updates in the whole system.
Then one can reformulate \eqref{eq:Thm:OpenSystems:ConvRate:generalFi} in terms of $\rho_R$ using the fact that $p_U = \tfrac{1}{1+\rho_R}$, and it follows that convergence is guaranteed as long as
\begin{align}
    \label{eq:OpenSystems:ConvGuarantee:rhoR}
    \rho_R < \frac{1}{(n-1)\kappa},
\end{align}
namely as long as on average at most one replacement happens every $(n-1)\kappa$ RCD updates.
This is equivalently formulated in terms of $p_a$ and $p_e$ which we remind respectively denote the probability that at an event a particular agent is replaced and a particular pair of agents performs a RCD update (see \eqref{eq:OpenSystems:pepaWRTpUpR}), and it follows that 
$p_a < \tfrac{1}{2\kappa}p_e$. 

Observe moreover that the recurrence equation \eqref{eq:Thm:OpenSystems:ConvRate:generalFi} can be solved, yielding
\begin{equation}
    \label{eq:EConvRate:Open:CkwrtC0}
    \E\brk{C^k} - \gamma 
    \leq \prt{1 + \frac{\rho_R - \tfrac{1}{(n-1)\kappa}}{1+\rho_R}}^k\prt{\E\brk{C^0} -\gamma},
\end{equation}
where 
\begin{equation}
    \label{eq:EConvRate:Open:CkwrtC0:Gamma'}
    \gamma 
    = 8n\kappa\frac{\prt{1+\tfrac{1}{\sqrt{\kappa}} + \tfrac{\abs{b}}{n}}\rho_R}{\tfrac{1}{(n-1)\kappa}-\rho_R},
\end{equation}
so that provided convergence occurs, there holds
$$
\lim_{k\to\infty}\E\brk{C^k}\le\gamma.
$$

Observe that conservatism is induced by the term $\rho_R$ in the numerator of the contraction rate of \eqref{eq:EConvRate:Open:CkwrtC0}. 
It exhibits how replacements can get in the way of convergence.
In particular, as $\rho_R$ increases, $E\brk{C^k}$ is expected to grow unbounded.
Conversely, with $\rho_R$ decreasing, it is expected that $\gamma\sim \rho_R(n\kappa)^2$, until $\gamma\to0$ as $\rho_R\to0$ (\emph{i.e.}, in total absence of replacements, or equivalently as $p_U\to1$). In that case, one retrieves the contraction rate of Proposition~\ref{Prop:EConvRate:updates}, and
$$
\E\brk{C^k} \leq \prt{1-\tfrac{1}{(n-1)\kappa}}^k\E\brk{C^0}.
$$

\subsection{Tightness Analysis}
\label{Sec:OpenSystems:QuadFun}

A critical part determining the tightness of our result is the analysis of the impact of a function change from Proposition~\ref{Prop:OpenSystems:minChange}, currently in $O(n\kappa)$.
That result is most likely conservative because it includes the possibility for all the functions to be replaced at once, whereas only replacements of single functions are allowed by our model.
In this section, we show why we expect a possible improvement of that result that does not scale with $n$, with two different approaches.

\paragraph{Quadratic functions}
We consider the particular case where every local objective function is quadratic, as defined in the following assumption.

\begin{assumption}
\label{Ass:QuadFun}
    For all $i$, there holds $f_i(x_i) = \theta_i(x_i-\mu_i)^2$, for some $\theta_i \in \frac12[\alpha,\beta]$, and for some $\mu_i\in[-1,1]$.
\end{assumption}

Under Assumption~\ref{Ass:QuadFun}, it is possible to obtain an alternative result for Proposition~\ref{Prop:OpenSystems:minChange} in $O(\kappa^6)$ that yields the following theorem that is proved in Appendix~\ref{Sec:Appendix:ProofQuadFun}.

\begin{theorem}
\label{Thm:OpenSystems:ConvRate:QuadFun}
    Under Assumptions~\ref{Ass:Statement:Fab} to \ref{Ass:QuadFun}, the iteration rule \eqref{eq:Statement:RCD:UpdateRule} applied on a system subject to replacements generates a sequence of estimates $x^k$ satisfying for all $k\geq0$
    \begin{align}
        \label{eq:Thm:OpenSystems:ConvRate:QuadFun}
        \E\brk{C^{k+1}} \leq \prt{2-p_U\prt{1+\tfrac{1}{(n-1)\kappa}}}\E\brk{C^k} + \Gamma',
    \end{align}
    with
    \begin{small}
    \begin{equation*}
        \label{eqThm:OpenSystems:ConvRate:QuadFun:Gamma}
        \Gamma' = (1-p_U)8\prt{\tfrac{\kappa^3+\kappa n-2}{\kappa n}+\tfrac{(\vert b\vert+n)^2(\kappa-1)^2\kappa^2\prt{\kappa^2n^2+n-1}}{n^4}}.
    \end{equation*}
    \end{small}
\end{theorem}

The difference between Theorems~\ref{Thm:OpenSystems:ConvRate:generalFi} and \ref{Thm:OpenSystems:ConvRate:QuadFun} lies in the terms $\Gamma$ and $\Gamma'$, which are respectively in $O(\kappa n)$ and $O(\kappa^6)$.
That difference illustrates the possible improvement achievable for our bound with respect to $n$ at the cost of its tightness in $\kappa$. 
Fig.~\ref{Fig:pU_plot} presents the results of the computations of $\E\brk{C^k}$ based on 10000 realizations of the process and the upper bound given by \eqref{eq:Thm:OpenSystems:ConvRate:QuadFun} for a network constituted of $n=5$ agents, with $\kappa=1.2$, $p_U=0.95$ and $b=1$.
The figure seems to confirm the tightness of the convergence rate derived for quadratic functions provided that $\kappa$ is not too large.

\begin{figure}
\centering
\includegraphics[width=8cm]{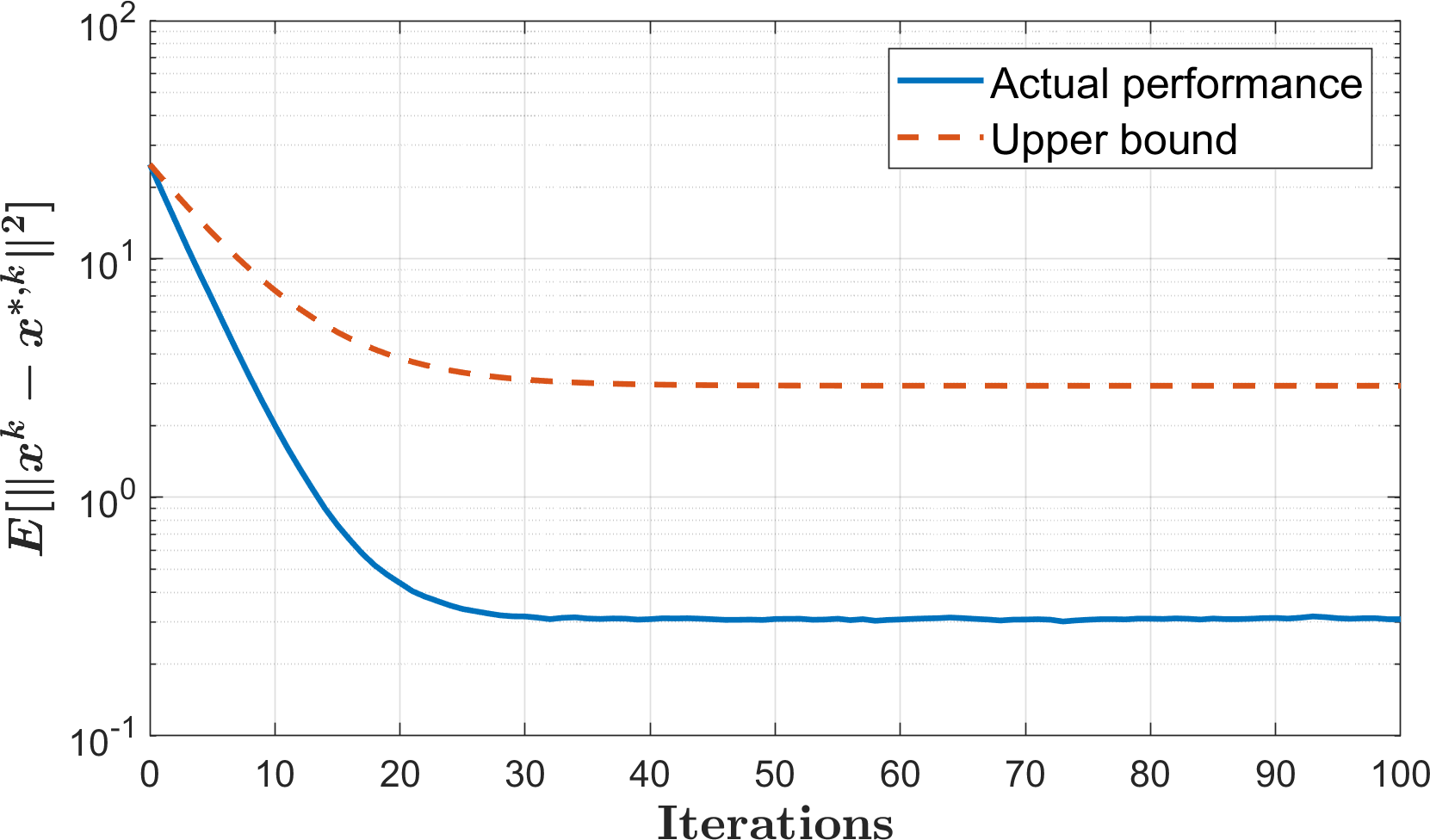}
\caption{Performance of the RCD algorithm in an open system of $5$ agents with $\kappa=1.2$, $b=1$ and $p_U=0.95$ (\textit{i.e.}, $\rho_R\approx0.053$), where each local objective function is quadratic.
The plain blue line represents the actual performance of the algorithm, and the dashed red line the upper bound \eqref{eq:Thm:OpenSystems:ConvRate:QuadFun} obtained from Theorem~\ref{Thm:OpenSystems:ConvRate:QuadFun}. The expected value was computed with 10000 realizations of the process. 
}
\label{Fig:pU_plot}
\end{figure}

\paragraph{PESTO analysis}
The possibility to improve our bound is also illustrated by an analysis performed using the PESTO toolbox \cite{PESTO}, which allows deriving numerical exact bounds for questions related to convex functions.
Using PESTO, we obtain an upper bound on a generalization of $\norm{x^{(2)}-x^{(1)}}^2$, where $x^{(1)}$ and $x^{(2)}$ are defined in \eqref{eq:OpenSystems:x+&x-}, for multi-dimensional functions $f_i$. 
Details on the way the analysis with PESTO was performed are presented in Appendix~\ref{Sec:Appendix:PESTO}.

The results of the PESTO analysis are presented in Fig.~\ref{fig:PESTO&conjecture}, and suggest a sublinear increase of the bound with $n$ for some fixed $\kappa$, and with $b=1$.
Additional numerical exploration of that result suggests a possible asymptotic independence of the bound with respect to $n$ for fixed values of $\kappa$ and with $b=1$ (similar results were observed for other values of $b$), and we conjecture the following bound, also illustrated in Fig.~\ref{fig:PESTO&conjecture}:
\begin{equation}
    \label{eq:conjecture}
    \norm{x^{(2)}-x^{(1)}}^2
    \leq  (\kappa+1)^2-\frac{c_1\kappa^3}{n+\kappa+c_2},    
\end{equation}
for some $c_1,c_2\in\R$.
The above conjecture would yield an equivalent result as that of Theorem~\ref{Thm:OpenSystems:ConvRate:QuadFun} with 
\begin{equation*}
    \Gamma' = 2(1-p_U)\prt{(\kappa+1)^2 - \frac{c_1\kappa^3}{n+\kappa+c_2}}.
\end{equation*}
Interestingly, whereas $\Gamma'$ grows in $\kappa^2$ for most values of $c_1$ and $c_2$, some choices yield a linear growth of $\Gamma'$ in $\kappa$ (\textit{e.g.}, if $c_1=1$, as shown in Fig.~\ref{fig:PESTO&conjecture}). 
Moreover, $\Gamma'$ does not grow with $n$ anymore, consistently with the improvement that we expect to achieve for future work.

\begin{figure}
    \centering
    \includegraphics[scale=0.5,clip = true, trim=2cm 11cm 2cm 11cm,keepaspectratio]{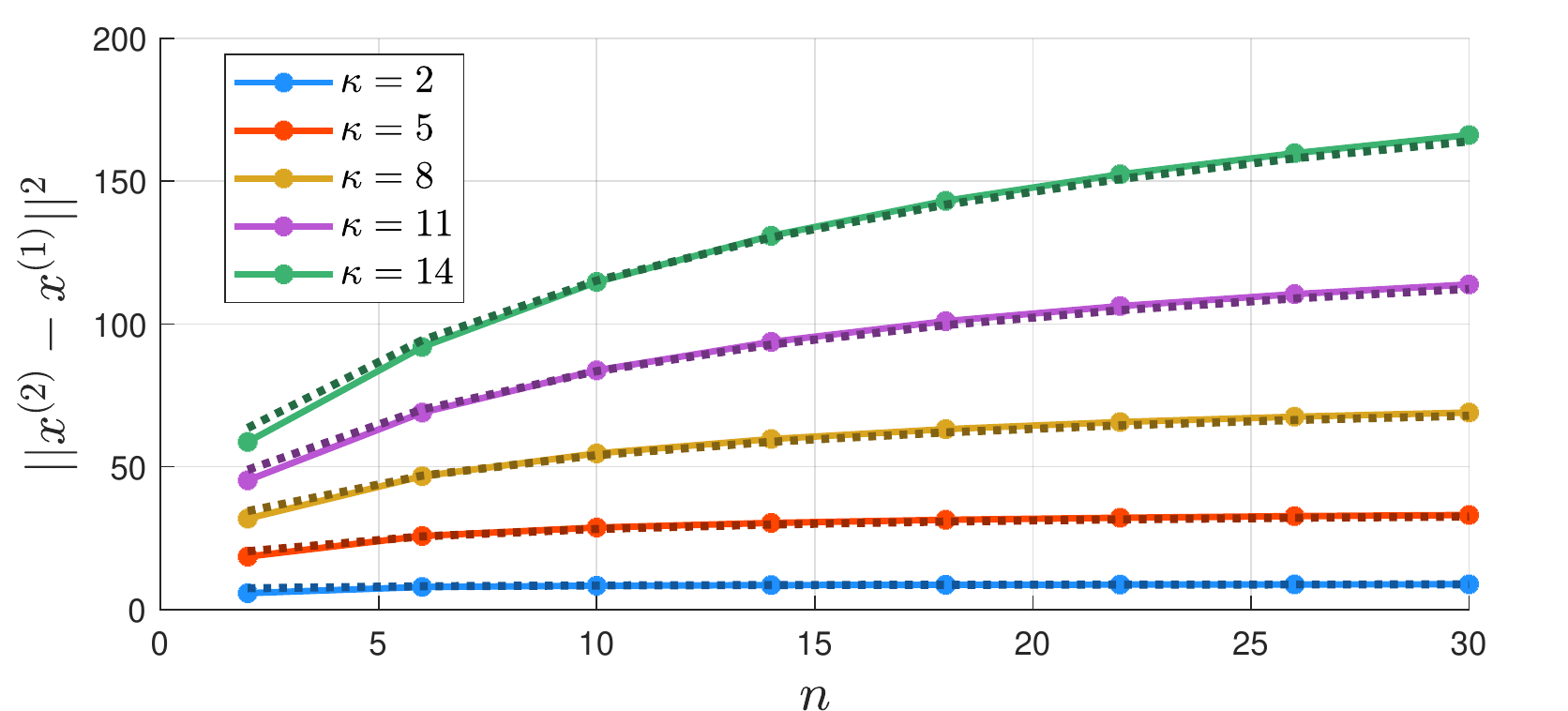}
    \caption{Evolution of the exact worst-case of $\norm{x^{(2)}-x^{(1)}}^2$ with $x^{(1)}$, $x^{(2)}$ as defined in \eqref{eq:OpenSystems:x+&x-}, obtained using PESTO with respect to $n$ for different values of $\kappa$, and with $b=1$ (plain line). The results are compared with the conjecture \eqref{eq:conjecture} with $c_1=1$ and $c_2=1$ (dotted line).
    }
    \label{fig:PESTO&conjecture}
\end{figure}

\section{Conclusion}

In this work we analyzed the random coordinate descent algorithm for a complete graph in an open multi-agent systems scenario when agents can be replaced during the iterations. 
We analyzed the behavior of the minimizer under replacement events, and derived an upper bound for the error in expectation and conditions for its stability.

As future work, we would like to improve the bounds for general classes of functions following the discussion on tightness performed in Section~\ref{Sec:OpenSystems:QuadFun}, especially since tighter bounds were obtained for particular settings and can be conjectured empirically.
Possible extensions include considering agents interacting through networks with different graph topologies, and generalizing the constraint to general $a\in\R^n$. 
Also, it would be interesting to consider the case where the states of the agents in the network are $d$-dimensional and where more than one edge can be updated at each iteration. 

\bibliographystyle{IEEEtran}
\bibliography{CDC2021_FINAL.bib}

\appendix

\subsection{Proof of Theorem~\ref{Thm:OpenSystems:ConvRate:QuadFun}}
\label{Sec:Appendix:ProofQuadFun}

\begin{proposition}
\label{Prop:QuadBound}
Under Assumptions~\ref{Ass:Statement:Fab}, \ref{Ass:Statement:B(0,1)}, \ref{Ass:Statement:a=1} and \ref{Ass:QuadFun}, for $x^{(1)}$ and $x^{(2)}$ as defined in \eqref{eq:OpenSystems:x+&x-}, there holds
\begin{align}
    \norm{x^{(2)}-x^{(1)}}^2
    &\le 8\prt{\dfrac{\kappa^3+\kappa n-2}{\kappa n}}\label{eq:Prop:QuadBound:difference}\\
    &+\dfrac{8(\vert b\vert+n)^2(\kappa-1)^2\kappa^2}{n^4}\prt{\kappa^2n^2+n-1}\nonumber.
\end{align}
\end{proposition}

\begin{proof} 
Using Lagrange multipliers, we have:
$
\mathcal L(x,\lambda) = f(x) + \lambda(\mathds 1^\top x-b),
$
which yields
\begin{align*}
    \frac{\partial \mathcal L}{\partial x}
    &= \nabla f(x) + \lambda \mathds 1_n = 0
    \implies 
    x_i = \mu_i - \frac{\lambda}{2\theta_i};\\
    \frac{\partial \mathcal L}{\partial \lambda}
    &= \mathds 1_n^\top x-b = 0
    \implies \sum\nolimits_i \mu_i - \lambda \sum\nolimits_i \frac{1}{2\theta_i} = b,
\end{align*}
and hence the minimizer of each local function in \eqref{eq:Statement:Objective:OpenResourceAllocProblem} is given by
$
x_i^*= \mu_i + \frac{b - \sum_j\mu_j}{\sum_j\theta_i/\theta_j}.
$
The difference between the minimizers can be expressed as:

\vspace{-0.5cm}
\begin{small}
\begin{equation}\label{eq:norm_difference}
\norm{x^{(1)}-x^{(2)}}^2=\sum_{i=1}^{n-1}\prt{x^{(1)}_i-x^{(2)}_i}^2+\prt{x^{(1)}_n-x^{(2)}_n}^2.
\end{equation}
\end{small}
\vspace{-0.5cm}

To find an upper bound for \eqref{eq:norm_difference}  let consider first the difference between the minimizers when only $\theta_n$ changes. Let denote $M_n=\sum_{i=1}^{n-1}\mu_i$ and $\zeta_0=\sum_{j}^{n-1}1/\theta_j$, such that
$$
\zeta^{(1)}=\sum\nolimits_j^{n-1}\dfrac{1}{\theta_j}+\dfrac{1}{\theta_n^{(1)}}=\zeta_0+\dfrac{1}{\theta_n^{(1)}}
$$
$$
\zeta^{(2)}=\sum\nolimits_j^{n-1}\dfrac{1}{\theta_j}+\dfrac{1}{\theta_n^{(2)}}=\zeta_0+\dfrac{1}{\theta_n^{(2)}}.
$$

For $i\ne n$ we have:
\begin{align*}
    \prt{x_i^{(1)} - x_i^{(2)}}^2
    &= \prt{\mu_i + \tfrac{b-M_n-\mu_n}{\theta_i\zeta^{(1)}} - \mu_i - \tfrac{b-M_n-\mu_n}{\theta_i\zeta^{(2)}}}^2\\
    &=\tfrac{(b-M_n-\mu_n)^2}{\theta_i^2}\prt{\tfrac{\zeta^{(2)}-\zeta^{(1)}}{\zeta^{(1)}\zeta^{(2)}}}^2
\end{align*}
Since $\theta_i\in\frac{1}{2}[\alpha,\beta]$ we obtain:    
\begin{align}
    \prt{x_i^{(1)} - x_i^{(2)}}^2
    &\leq \tfrac{(b\!-\!M_n\!-\!\mu_n)^2}{\theta_i^2}4\prt{\tfrac{\beta\!-\!\alpha}{\alpha\beta}}^2\!\prt{\tfrac{1}{\zeta^{(1)}\zeta^{(2)}}}^2\nonumber\\
    &\leq 4\tfrac{(\vert b\vert+n)^2}{\alpha^2}\tfrac{(\kappa\alpha-\alpha)^2}{\alpha^2\beta^2}\tfrac{1}{n^4\beta^{-4}}\nonumber\\
    &=\tfrac{4(\vert b\vert+n)^2(\kappa-1)^2\kappa^2}{n^4}.\label{eq:i_neq_n}
\end{align}
For $i=n$ we have:
\begin{align}
    \prt{x_n^{(1)} - x_n^{(2)}}^2
    &= \prt{\mu_n + \tfrac{b-M_n-\mu_n}{\theta_n^{(1)}\zeta^{(1)}} - \mu_n - \tfrac{b-M_n-\mu_n}{\theta_n^{(2)}\zeta^{(2)}}}^2\nonumber\\
    &=(b-M_n-\mu_n)^2\prt{\tfrac{\theta_n^{(2)}\zeta^{(2)}-\theta_n^{(1)}\zeta^{(1)}}{\theta_n^{(1)}\theta_n^{(2)}\zeta^{(1)}\zeta^{(2)}}}^2\nonumber\\
    &=\frac{(b-M_n-\mu_n)^2\prt{\theta_n^{(2)}-\theta_n^{(1)}}^2\zeta_0^2}{\prt{\theta_n^{(1)}\theta_n^{(2)}\zeta^{(1)}\zeta^{(2)}}^2}\nonumber\\
    &\leq\dfrac{4(\vert b\vert+n)^2(\beta-\alpha)^2\kappa^4n^2}{\alpha^2n^4}\nonumber\\
    &=\dfrac{4(\vert b\vert+n)^2(\kappa-1)^2\kappa^4}{n^2}.\label{eq:i_eq_n}
\end{align}
Now we consider the general case when both $\mu_n$ and $\theta_n$ can change. Let denote $\delta_\mu=\mu_n^{(1)}-\mu_n^{(2)}$. 

For $i\ne n$ we have:
\begin{align*}
    \prt{x_i^{(1)}\! -\! x_i^{(2)}}^2\!\!
    &= \prt{\mu_i + \tfrac{b-M_n-\mu_n^{(1)}}{\theta_i\zeta^{(1)}} - \mu_i - \tfrac{b-M_n-\mu_n^{(2)}}{\theta_i\zeta^{(2)}}}^2\\
    &=\prt{\tfrac{b-M_n-\mu_n^{(2)}}{\theta_i\zeta^{(1)}}-\tfrac{b-M_n-\mu_n^{(2)}}{\theta_i\zeta^{(2)}}-\tfrac{\delta_\mu}{\theta_i\zeta^{(1)}}}^2\\
    &\leq \prt{\left\vert \tfrac{b-M_n-\mu_n^{(2)}}{\theta_i\zeta^{(1)}}-\tfrac{b-M_n-\mu_n^{(2)}}{\theta_i\zeta^{(2)}} \right\vert+\left\vert \tfrac{\delta_\mu}{\theta_i\zeta^{(1)}} \right\vert}^2\\
    &\leq \!2\!\prt{\tfrac{b-M_n-\mu_n^{(2)}}{\theta_i\zeta^{(1)}}\!-\!\tfrac{b-M_n-\mu_n^{(2)}}{\theta_i\zeta^{(2)}}}^2\!\!+\!2\!\prt{\tfrac{\delta_\mu}{\theta_i\zeta^{(1)}}}^2.
\end{align*}
Then, by using \eqref{eq:i_neq_n} we obtain:
\begin{align}
    \prt{x_i^{(1)}\! -\! x_i^{(2)}}^2\!
    &\leq \dfrac{8(\vert b\vert+n)^2(\kappa-1)^2\kappa^2}{n^4}\!+\!2\prt{\frac{2}{\alpha \sum_j 1/\beta}}^2\nonumber\\
    &\leq \dfrac{8(\vert b\vert+n)^2(\kappa-1)^2\kappa^2}{n^4}+\dfrac{8\kappa^2}{n^2}.\label{eq:bound_i_neq_n}
\end{align}

For $i=n$ we have:
\begin{align*}
    \prt{x_n^{(1)}\! -\! x_n^{(2)}}^2
    \!&= \prt{\mu_n^{(1)} + \tfrac{b-M_n-\mu_n^{(1)}}{\theta_n^{(1)}\zeta^{(1)}} - \mu_n^{(2)} - \tfrac{b-M_n-\mu_n^{(2)}}{\theta_n^{(2)}\zeta^{(2)}}}^2\\
    &=\prt{\delta_\mu\!+\!\tfrac{b-M_n-\mu_n^{(2)}}{\theta_n^{(1)}\zeta^{(1)}}\!-\!\tfrac{b-M_n-\mu_n^{(2)}}{\theta_n^{(2)}\zeta^{(2)}}\!-\!\tfrac{\delta_\mu}{\theta_n^{(1)}\zeta^{(1)}}}^2\\
    &\leq \Big(\left\vert \tfrac{b-M_n-\mu_n^{(2)}}{\theta_n^{(1)}\zeta^{(1)}}-\tfrac{b-M_n-\mu_n^{(2)}}{\theta_n^{(2)}\zeta^{(2)}} \right\vert\\
    &\quad\; +\left\vert \delta_\mu- \tfrac{\delta_\mu}{\theta_n^{(1)}\zeta^{(1)}} \right\vert\Big)^2\\
    &\leq 2\prt{\tfrac{b-M_n-\mu_n^{(2)}}{\theta_n^{(1)}\zeta^{(1)}}-\tfrac{b-M_n-\mu_n^{(2)}}{\theta_n^{(2)}\zeta^{(2)}}}^2\\
    &\quad\;+2\prt{\delta_\mu-\tfrac{\delta_\mu}{\theta_n^{(1)}\zeta^{(1)}}}^2.
\end{align*}
Then, by using \eqref{eq:i_eq_n} we obtain:
\begin{align}
    \prt{x_n^{(1)} - x_n^{(2)}}^2\!\!
    &\!\leq\! \tfrac{8(\vert b\vert+n)^2(\kappa-1)^2\kappa^4}{n^2}\nonumber\\
    &\quad\! +8\prt{1-\tfrac{2}{\theta_n^{(1)}\zeta^{(1)}}+\tfrac{1}{(\theta_n^{(1)}\zeta^{(1)})^2}}\nonumber\\ 
    &\leq\! \tfrac{8(\vert b\vert+n)^2(\kappa-1)^2\kappa^4}{n^2}\!+\!8\prt{1\!-\!\tfrac{2}{\kappa n}\!+\!\tfrac{\kappa^2}{n^2}}.\label{eq:bound_i_eq_n}
\end{align}
Finally, using \eqref{eq:bound_i_neq_n} and \eqref{eq:bound_i_eq_n} in \eqref{eq:norm_difference} yields the conclusion.
\end{proof}

The proof of Theorem~\ref{Thm:OpenSystems:ConvRate:QuadFun} follows the same steps as that of Theorem~\ref{Thm:OpenSystems:ConvRate:generalFi} using Proposition~\ref{Prop:QuadBound} instead of Proposition~\ref{Prop:OpenSystems:EConvRate:repl}.

\subsection{PESTO implementation for the tightness analysis}
\label{Sec:Appendix:PESTO}

In this section we describe how the analysis relying on the PESTO toolbox presented in Section~\ref{Sec:OpenSystems:QuadFun} was performed.
This toolbox was initially developed to numerically compute the exact worst-case performance of first-order convex optimization algorithms, and more generally allows deriving exact bounds on questions related to convex functions.

We consider a general setting with multi-dimensional functions $f_i:\R^d\to\R$ that are $\alpha$-strongly convex and $\beta$-smooth, with $\argmin_x f_i(x) \in B(0,1)$.
We use PESTO to evaluate $\max \norm{x^{(2)}-x^{(1)}}^2$, where $x^{(1)}$ and $x^{(2)}$ are defined as in \eqref{eq:OpenSystems:x+&x-} for that multi-dimensional setting, and where we impose $\sum_i x_i^{(1)} = \sum_i x_i^{(2)} = v_b$, for some vector $v_b$ satisfying $\norm{v_b}=b$.
This setting is exactly equivalent to that of Proposition~\ref{Prop:OpenSystems:minChange} when $d=1$. 
However, PESTO does not allow imposing $d=1$ in the implementation, so that we solve this more general problem, whose solution will thus also be valid for more general values of $d$.

Hence, the problem is implemented in PESTO as $\max\norm{x^{(2)}-x^{(1)}}^2$, with \emph{the variables of the problem being the functions $f_i$, the vector $v_b$ and the decision variables $x_i$}, so that it derives the performance achieved by the empirical worst-case instance of the problem. 

The result that is obtained for $\kappa=5$ and $b=1$ with respect to $n$ is presented in Fig.~\ref{fig:PESTO}.
The top plot shows the sublinear increase of the bound with $n$, and suggests its possible asymptotic independence in $n$.
The bottom plot seems to confirm that conjecture, and it is expected that the bound converges to $(\kappa+1)^2$ as $n\to\infty$, consistently with the analysis and conjecture presented in Section~\ref{Sec:OpenSystems:QuadFun}.

\begin{figure}
    \centering
    \includegraphics[scale=0.6,clip = true, trim=3.5cm 8.5cm 2.5cm 8.5cm,keepaspectratio]{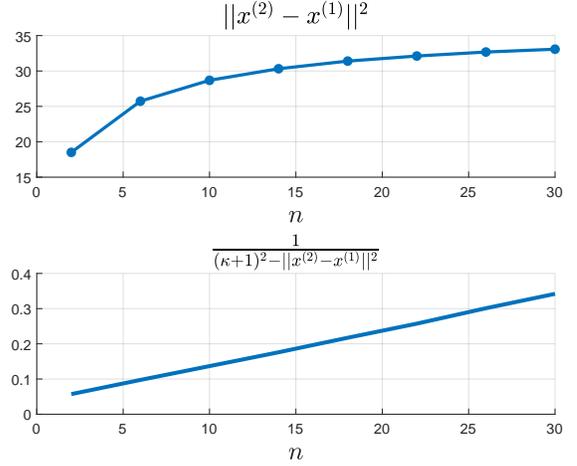}
    \caption{(Top): Evolution of the exact worst-case empirical performance for $\norm{x^{(2)}-x^{(1)}}^2$ with $\kappa=5$ and $b=1$ with respect to $n$. (Bottom): Variation of the result presented in (top), suggesting the possible asymptotic independence of the bound with $n$.}
    \label{fig:PESTO}
\end{figure}

\end{document}